\documentclass{acm_proc_article-sp}

\usepackage{algorithm}
\usepackage{algorithmic}
\usepackage{amsmath}
\usepackage{amsfonts}
\usepackage{color}
\usepackage{graphicx}
\usepackage{hyperref}
\usepackage{xspace}

\usepackage{tikz}
\usetikzlibrary{shapes}
\usetikzlibrary{shapes.multipart}
\usetikzlibrary{positioning}
\usetikzlibrary{decorations}
\usetikzlibrary{decorations.pathmorphing} 
\usetikzlibrary{fit}		
\usetikzlibrary{backgrounds}	
\usetikzlibrary{matrix,arrows}
\usetikzlibrary{calc}
\usetikzlibrary{through}

 \newtheorem{theorem}{Theorem}[section]
 \newtheorem{lemma}[theorem]{Lemma}
 \newtheorem{definition}[theorem]{Definition}
 \newtheorem{example}[theorem]{Example}

\author{
\alignauthor
Babak Bagheri Harari \hspace{+0.15in}
Val Tannen \\
\affaddr{Computer \& Information Science Department} \\
\affaddr{University of Pennsylvania} \\
}

\begin{document}

\newcommand{\reminder}[1]{\textcolor{red}{#1}}
\newcommand{\eat}[1]{}

\newcommand{\A}{\mathcal{A}} \newcommand{\B}{\mathcal{B}}
\newcommand{\C}{\mathcal{C}} \newcommand{\D}{\mathcal{D}}
\newcommand{\E}{\mathcal{E}} \newcommand{\F}{\mathcal{F}}
\newcommand{\G}{\mathcal{G}} \renewcommand{\H}{\mathcal{H}}
\newcommand{\I}{\mathcal{I}} \newcommand{\J}{\mathcal{J}}
\newcommand{\K}{\mathcal{K}} \renewcommand{\L}{\mathcal{L}}
\newcommand{\M}{\mathcal{M}} \newcommand{\N}{\mathcal{N}}
\renewcommand{\O}{\mathcal{O}} \renewcommand{\P}{\mathcal{P}}
\newcommand{\Q}{\mathcal{Q}} \newcommand{\R}{\mathcal{R}}
\renewcommand{\S}{\mathcal{S}} \newcommand{\T}{\mathcal{T}}
\newcommand{\U}{\mathcal{U}} \newcommand{\V}{\mathcal{V}}
\newcommand{\W}{\mathcal{W}} \newcommand{\X}{\mathcal{X}}
\newcommand{\Y}{\mathcal{Y}} \newcommand{\Z}{\mathcal{Z}}

\newcommand{\ansa}[3]{\ensuremath{\Gamma_{#1}^{#2}(#3)}\xspace}
\newcommand{\ansk}[2]{\ensuremath{\text{ans}^{#1}[#2]}\xspace}
\newcommand{\gat}[3]{\ensuremath{\Phi_{#1}^{#2}(#3)}\xspace}
\newcommand{\gatk}[3]{\ensuremath{\Phi_{#1}^{#2}[#3]}\xspace}

\newcommand{\ra}{\rightarrow}
\newcommand{\sq}{\bar}
\newcommand{\lt}{\leadsto}
\newcommand{\yield}[2]{\Rightarrow_{#1}^{#2}}
\newcommand{\pre}{\text{pre}}
\newcommand{\post}{\text{post}}
\newcommand{\gaifman}[1]{\ensuremath{\G}(#1)\xspace}

 \newcommand{\entails}{\ensuremath{\models}\xspace}

\newcommand{\tup}[1]{\langle#1\rangle}            
\newcommand{\lisom}[2]{\ensuremath\equiv^{#1}_{#2}\xspace}            

\newcommand{\yieldAdd}[1]{\ensuremath\text{toAdd}(#1)\xspace}            
\newcommand{\yieldDell}[1]{\ensuremath\text{toDell}(#1)\xspace}

\newcommand{\ucqn}{UCQ\ensuremath{^{\neg}}\xspace}
\newcommand{\gaq}[2]{\text{GA}_{#1}^{#2}\xspace}

\newcommand{\add}{\textbf{add}}
\newcommand{\del}{\textbf{del}}
\newcommand{\set}[1]{\{#1\}}
\newcommand{\ews}[1]{\ensuremath{\eta[#1]}}
\newcommand{\ans}[1]{\ensuremath{\textit{ans}(#1)\xspace}}
\newcommand{\support}[1]{\ensuremath{\text{support}(#1)}\xspace}

\newcommand{\const}[1]{\textit{const}(#1)}

\newcommand{\contd}{\ensuremath{\textit{contd}}\xspace}
\newcommand{\comp}[1]{\ensuremath{\textit{comp}(#1)}\xspace}
\newcommand{\relat}[1]{\ensuremath{\textit{rel}(#1)}\xspace}
\newcommand{\dleq}{\ensuremath{\textit{ :- }}\xspace}
\newcommand{\cooc}[1]{\ensuremath{\textit{coocc(#1)}}\xspace}
\newcommand{\adom}[1]{\ensuremath{\textit{adom}(#1)}\xspace}

\newcommand{\homo}[1]{~\ensuremath{\ra^{#1}}~}
\newcommand{\equivv}[2]{\ensuremath{\textit{equiv}(#1, #2)}\xspace}
\newcommand{\flip}[1]{\ensuremath{\textit{flip}(#1)}\xspace}

\newcommand{\counts}[1]{\ensuremath{\textit{count}(#1)}\xspace}

 \newcommand{\nterms}[1]{\ensuremath{\textit{terms}(#1)}\xspace}

\newcommand{\idfunction}[1]{\ensuremath{\textit{id}(#1)}\xspace}

\newcommand{\domain}[1]{\ensuremath{\textsc{dom}(#1)}\xspace}
\newcommand{\image}[1]{\ensuremath{\textsc{im}(#1)}\xspace}
\newcommand{\CONST}{\C}
\newcommand{\TERMS}{\T}
\newcommand{\VARIABLES}{\V}
\newcommand{\arity}[1]{\textsc{arity}(#1)\xspace}
\newcommand{\HERBRAND}{\ensuremath{\mathbb{U}}\xspace}
\newcommand{\IDRel}{\textit{id}}

\newcommand{\ISM}[1]{\equiv_{#1}}
\newcommand{\ISMA}[2]{\equiv_{#2}^{#1}}

\newcommand{\HMP}[1]{ \triangleright_{#1}}
\newcommand{\HMPA}[2]{ \triangleright_{#2}^{#1}}
\newcommand{\true}{\mathsf{true}}
\newcommand{\false}{\mathsf{false}}

\newcommand{\disting}[1]{\ensuremath{\textit{disting}(#1)}\xspace}
\newcommand{\nondisting}[1]{\ensuremath{\textit{nondisting}(#1)}\xspace}
\newcommand{\var}[1]{\ensuremath{\textit{var}(#1)}\xspace}
\newcommand{\ineq}[1]{\mathit{\textit{ineq}(#1)}\xspace}

\newcommand{\deq}{\ensuremath{~\textit{:-}}~}
\newcommand{\coqu}[3]{\ensuremath{#1(#2,~\textit{c}(#3)) \deq\xspace}}


\newcommand{\myi}{\emph{(i)}\xspace}
\newcommand{\myii}{\emph{(ii)}\xspace}
\newcommand{\myiii}{\emph{(iii)}\xspace}
\newcommand{\myiv}{\emph{(iv)}\xspace}
\newcommand{\myv}{\emph{(v)}\xspace}
\newcommand{\myvi}{\emph{(vi)}\xspace}

\newcommand{\extn}[1]{\text{ext}(#1)\xspace}
\newcommand{\core}[1]{\text{core}(#1)\xspace}

\title{Decidability of Equivalence of Aggregate Count-Distinct
  Queries}

\maketitle

\abstract{We address the problem of equivalence of count-distinct
  aggregate queries, prove that the problem is decidable, and can be
  decided in the third level of Polynomial hierarchy. We introduce the
  notion of core for conjunctive queries with comparisons as an
  extension of the classical notion for relational queries, and prove
  that the existence of isomorphism among cores of queries is a
  sufficient and necessary condition for equivalence of conjunctive
  queries with comparisons similar to the classical relational
  setting. However, it is not a necessary condition for equivalence of
  count-distinct queries. We introduce a relaxation of this condition
  based on a new notion, which is a potentially new query equivalent
  to the initial query, introduced to capture the behavior of
  count-distinct operator.  }

\section{Introduction}

The problem of deciding equivalence among conjunctive aggregate
count-distinct queries has been investigated starting from late 90s,
and the problem is known to be decidable for most of the aggregate
operators \cite{CoSN05,NuSS98,Cohen05,CoNS07,NuSS98} including SUM,
Average, and COUNT. For for count-distinct queries, although some
sufficient conditions have been proposed \cite{NuSS98,CoNS07}, it has
still been open if the problem is decidable or not.  In this paper we
study the problem of equality of count-distinct conjunctive queries,
and provide a sufficient and necessary condition for the
problem. First, we introduce the notion of core for conjunctive
queries with comparisons. We show that similar to the classic results
for relational conjunctive queries, the existence of an isomorphism
between cores of conjunctive queries with comparisons is a sufficient
and necessary conditions for their equality. While this condition
provides us with a sufficient condition for the equality of
conjunctive queries, it is easy to show that existence of isomorphism
is not a necessary condition for equivalence of count-distinct
queries. The necessary and sufficient condition is obtained by a
relaxation of the above condition based on a newly introduced notion
of flip of a query. The flip of a count-distinct query is a
potentially different count-distinct query equal to the original
query, obtained by flipping the direction of some of the comparisons
among the variables of the query. We show that given two queries to
compare, they are equal if and only if the first query or its flip is
isomorphic to the second query or its flip. Both computing of the core
of a query and its flip can be done in $\Delta^{\pi}_3$, which makes
an upper-bound for the complexity of our suggested algorithm.


\section{Preliminaries}

We consider an ordered dense domain $\Delta$, and a countably infinite
set $\VARIABLES$ of \emph{variables}. \emph{Terms} $\T$ are either
constants from $\Delta$ or variables, i.e., $\T=\Delta\cup\VARIABLES$.

\bigskip
\noindent{\bf Relations.} 
Given $n$ possibly different sets $T_1, \dots, T_n\subseteq\T$ of
terms, an \emph{n-ary relation} $R/n$ between them is a subset of
their cartesian products: $R\subseteq T_1 \times \dots \times T_n$.
We might simply use the term \emph{relation}, when the arity is
implicit or not relevant.
The \emph{identity relation} over a given set $T\subseteq\TERMS$,
denoted as $\IDRel(T)$ is the bijection relation from $T$ to $T$ that
maps all the terms to themselves, i.e., $\IDRel(t)= t$ for all
$t\in T$.
Given a binary relation $R: S\rightarrow S'$, we denote with
$\domain{R}$ the domain of $R$, i.e., the set of elements in $S$ on
which $R$ is defined.
We denote with $\image{R}$ the set of elements $s'$ in $S'$ such that
$(s,s')\in R$ for some $s\in S$.
A \emph{function (mapping)} $f: T_1 \mapsto T_2$ is a binary relation between
$T_1$ and $T_2$ such that for all terms $\{a, b, c\}\subseteq \T$ such
that $(a,b)\in f$ and $(a,c)\in f$, $b = c$. For functions, we might
use $f(a)=b$, instead of $(s, b)\in f$.
A \emph{bijection} $h$ between two sets $T_1$ and $T_2$ is a function
from $T_1$ to $T_2$, such that for all $a\in T_1$ there is exactly one
$b\in T_2$ such that $h(a)=b$, and for all $b\in T_2$ there is exactly
one $a\in T_1$ such that $h(a)=b$. A \emph{partial bijection} between
two sets $\T_1$ and $\T_2$ is a bijection between a subset of $\T_1$ 
and a subset of $\T_2$.  Given a bijection $h$, the inverse of $h$ denoted as
$h^{-1}$ is the function for which for all pairs $a$ and $b$,
$(a,b)\in h$ if and only if $(b,a) \in h^{-1}$.
Given a binary relation $R$ and a set $S$, the \emph{restriction} of
$R$ to $S$, denoted as $R|_S$ is the relation obtained by restricting
the domain of $R$ to the elements of $S$ as follows:
$\{(t_1,~t_2)~|~(t_1,~t_2)\in R$, and $t_1 \in S\}$.
We say a function $f'$ \emph{extends} the function $f$, if $f'$ is a
function, and $f\subseteq f'$.

\bigskip
\noindent{\bf Databases.}
We assume a countably infinite set $\R$ of \emph{relation names}, and to each
relation name $R\in\R$ assign an arity $\arity{R}$ greater than or
equal to zero.
A \emph{relational schema} is a finite set of relation names with
specified arities, i.e., is a finite subset of $\R$. Given a relation
name $R\in \R$, an instance of $R$ over the set $I\subseteq\T$ is a
finite subset of $I^{arity(R)}$.  Given a relational schema $S=\{R_1,
\dots, R_n\}$, a \emph{relational instance of $S$} over $I\subseteq\T$
is a set of instances of $R_1, \dots, R_n^{I}$.
Given a database instance $\I$, its
\emph{active domain} $\adom{\I}$ is the subset of $\T$ such that
$u\in\adom{\I}$ if and only if $u$ occurs in $\I$.
\bigskip

\noindent{\bf Assignment.}
An \emph{assignment} is a partial function $\theta:\T \mapsto \T$, such as
$\{t_1 \ra c_1, \dots, t_n \ra c_n\}$ which assigns the constants
$c_1, \dots, c_n \subseteq \Delta$ to the terms $t_1,\dots, t_n$.
Given a relational instance $I$ and an assignment $\theta$, we use the
notations $I\theta$ and $\theta(I)$ to denote the instance obtained
by applying the substitution $\theta$ to the terms of $I$. Formally,
\[
\begin{array}{rlll}
  I\theta = \theta(I) &= \{  &R(c_1\ldots, c_n) \mid &  R(t_1,\ldots,t_n)\in I  \mbox{, and } \\
  &             & &(t_i \ra c_i) \in \theta \\
  &             &  & \ \ \mbox{ for } i\in\{1,\ldots,n\}\}
\end{array}
\]

\bigskip
\noindent{\bf Homomorphism and Isomorphism} 
Given a bijection $\sigma$, and two relational instances $I$ and $J$,
a $\sigma$-homomorphism $h$ from $I$ to $J$ is a mapping from the
$\adom{I}$ to $\adom{J}$, defined as follows:
\begin{enumerate}
\item $h$ is a function $\adom{I}\mapsto\adom{J}$,
\item $h$ extends $\sigma$, i.e., $\sigma\subseteq h$,
\item for all relations names $R$ with arity $n$, $R(a_1, \dots,
  a_n)\in I$ implies $R(h(a_1), \dots,$ $h(a_n))\in J$.
\end{enumerate}
We say that there is a homomorphism from $I$ to $J$ wrt the bijection
$\sigma$, denoted as $I\HMPA{\sigma}{}J$ in case such a homomorphism
relation exists. We refer to $h$ as a witness of the homomorphism from
$I$ to $J$. We might use the notation $I\HMPA{\sigma}{h}J$ when needed
to emphasize that $h$ is a witness of the homomorphism from $I$ to
$J$. In case $\sigma$ is the bijection relation $\IDRel(C)$, we might
simply use $C$ to refer to the bijection relation $\sigma$.  When it
is non-relevant or clear from the context, we might drop $\sigma$, and
simply say there is a homomorphism from $I$ to $J$, denoted as
$I\HMPA{}{} J$.

We say that $I$ is $\sigma$-isomorphic to $J$, denoted as
$I\ISMA{\sigma}{}J$ if there exists a bijection relation $h$, such
that $I\HMPA{\sigma}{h}J$ and $J\HMPA{\sigma}{h^{-1}}I$. When it is
clear from the context, we might drop $\sigma$, and simply say $I$ and
$J$ are isomorphic, denoted as $I\ISMA{}{}J$. We also say that $h$ is
a witness of the \emph{$\sigma$-isomorphism} from $I$ to $J$, and when
it is clear from the context we might drop $\sigma$, and simply say
that $h$ is a witness of the isomorphism from $I$ to $J$.

\smallskip

\noindent{\bf Conjunctive Queries with Comparisons.} Here we borrow
the definition of \cite{Klug88} with following naming conventions. A
$q$ over the schema $\S$ can be specified with
\myi a set $\disting{q}=\{x_1,\dots, x_m\}$ of
distinguished variables, the sequence being called the \emph{summary
  row}, or just the summary; 
\myii A set $\nondisting{q}=\{y_1, \dots,
y_n\}$ of \emph{non-distinguished} (existentially quantified)
variables;
\myiii A set $\relat{q}={c_1, \dots, c_l}$ of distinct conjuncts,
each conjuncts $c_i$ being an atomic formula of the form $R(z_1,
\dots, z_r)$, where $R/r\in \S$ and each $z_i$ is a variable or a
constant, i.e., $z_i\in \disting{q}\cup\nondisting{q}\cup\Delta$; 
\myiv A set $\comp{q}=\{l_1, \dots, l_u\}$ of comparisons in the form
$z_i \sigma z_j$ in which $z_i, z_j\in
\disting{q}\cup\nondisting{q}\cup\Delta$, and $\sigma\in \{=, <,
\le\}$.  We denote the set of all distinguished and non-distinguished
variables of $q$ with $\var{q}$, i.e.,
$\var{q}=\disting{q}\cup\nondisting{q}$.

Given a relational instance $I$, and a conjunctive query $q$, the
answer of $q$ over $I$, denoted as $\ans{q, I}$ is a set of
assignments in the form $\sigma: \nondisting{q}\mapsto \adom{I}$ such
that for each  of the assignment $\sigma$
there exists an assignment $\sigma':\disting{q}\mapsto\adom{I}$, and for
$\sigma''=\sigma\cup \sigma'$ followings hold:
\begin{enumerate}
\item for each $c_i\in\relat{q}$, $\sigma''(c_i)\in I$
\item for each $l_i = {z_i\theta z_j}\in\comp{q}$, $\sigma''(z_i)
  \theta \sigma''(z_j)$.
\end{enumerate}

Given a query $\phi$ and an assignment $\theta$, we use the notation
$\phi\theta$ to denote the query obtained by applying $\theta$ to the
variables of $\phi$.




\bigskip
\noindent{\bf Homomorphism and Isomorphism of Conjunctive Queries with
Comparisons.} 
Given a bijection $\sigma$, and queries \\$\coqu{q_1}{\bar{x}}{y}$
$A_1(\bar{x}, {y}, z) \land C_1(\bar{x}, {y}, z)$, and
$\coqu{q_2}{\bar{x}}{y}$ $A_2(\bar{x}, {y}, z) \land C_2(\bar{x}, {y},
z)$, a $\sigma$-homomorphism $h$ from $q_1$ to $q_2$ is a mapping such
that $A_1 \HMPA{\sigma'}{h}A_2$ for $\sigma'=\sigma \cup
(\bar{x}\mapsto\bar{x}$), and for all linearization $L$ of $C_2$, for
all comparison $(r ~\sigma~ s) \in C_1$, $L\models h(r) ~\sigma~ h(s)$.

We say that there is a homomorphism from $q_1$ to $q_2$ wrt the
bijection $\sigma$, denoted as $q_1\HMPA{\sigma}{}q_2$ in case such a
homomorphism relation exists. We refer to $h$ as a witness of the
homomorphism from $q_1$ to $q_2$. We might use the notation
$q_1\HMPA{\sigma}{h}q_2$ when needed to emphasize that $h$ is a
witness of the homomorphism from $q_1$ to $q_2$. In case $\sigma$ is
the bijection relation $\IDRel(C)$, we might simply use $C$ to refer
to the bijection relation $\sigma$.  When it is non-relevant or clear
from the context, we might drop $\sigma$, and simply say there is a
homomorphism from $q_1$ to $q_2$, denoted as $q_1\HMPA{}{} q_2$.

We say that $q_1$ is $\sigma$-isomorphic to $q_2$, denoted as
$q_1\ISMA{\sigma}{}q_2$ if there exists a bijection relation $h$, such
that:
\begin{enumerate}
\item $h$ extends $\sigma$, i.e., $\sigma\subseteq h$,
\item for all relations names $R$ with arity $n$, $R(a_1, \dots,
  a_n)\in I$ implies $R(h(a_1), \dots,$ $h(a_n))\in J$.
\item for all relations names $R$ with arity $n$, $R(a_1, \dots,
  a_n)\in J$ implies $R(h^{-1}(a_1), \dots,$ $h^{-1} (a_n))\in I$;
\item for all comparisons $r~\sigma~s\in C_1$, $h(r)~\sigma~h(s)\in
  C_2$;
\item for all comparisons $r~\sigma~s\in C_2$, $h^{-1}(s)\in C_1.$
\end{enumerate}
When it is clear from the context, we might drop $\sigma$, and simply
say $q_1$ and $q_2$ are isomorphic, denoted as $q_1\ISMA{}{}q_2$. We
also say that $h$ is a witness of the \emph{$\sigma$-isomorphism} from
$q_1$ to $q_2$, and when it is clear from the context we might drop
$\sigma$, and simply say that $h$ is a witness of the isomorphism from
$q_1$ to $q_2$.


\bigskip
\noindent{\bf Extension of a Query.} 
Given a count-distinct query $\coqu{q}{\bar{x}}{y}$ $A(\bar{x}, {y},
z) \land C(\bar{x}, {y}, z)$ the extension of $q$ denoted as
$\extn{q}$ is the count distinct query $\coqu{q'}{\bar{x}}{y}$
$A(\bar{x}, {y}, z) \land C'(\bar{x}, {y}, z)$, in which $C'$ is a set
of comparisons defined as follows:
\[
C' = \{ r~\sigma~s ~|~ r, s \in \adom{q} \text{ and } C\models
r~\sigma~s\}.
\]

\begin{example} Lets assume $\coqu{q}{\bar{x}}{y}$ $A(x)\land A(y)
  \land B(z)$ $ \land x<y \land y<z$. Then the extension of $q$ is
  $\coqu{q}{\bar{x}}{y}$ $A(x)\land A(y) \land B(z)$ $ \land x<y \land
  y<z \land x<z$.
\end{example}

It is easy to see that the extension of a query always exists and is
homomorphically equivalent to the original query.




\section{Equivalence of Count-Distinct Queries}
In this section we address our main results on equivalence of
count-distinct queries.We first introduce the notion of core of
conjunctive queries with comparison, and show that existence of
isomorphism among core of queries coincides with their equality. Then
we will introduce a relaxation of this condition through the notion of
flip of conjunctive queries, which is introduced in this paper to
capture when two non-isomorphic count-distinct queries can be
equal. Finally, we will prove that two count-distinct queries are
equal if and only if at least one of the pairs of the original query
or their flips are isomorphic.

\subsection{Equivalence of Conjunctive Queries with Comparison}
Here we introduce the core of conjunctive queries with comparions as a
generalization of the concept of core for relational conjunctive
queries. We will show that our notion of core coincides with the
classic definition when restricted to the relational
queries. Moreover, similar to the classic notion of core, existence of
isomorphism among core of queries with comparisons implies equality of
the queries, and vice versa.

\begin{definition}[Core] Given a conjunctive query with comparisons
  $\coqu{q}{\bar{x}}{y}$ $A(\bar{x}, {y}, z) \land C(\bar{x}, {y},
  z)$, a core of $q$ is a query such as $\coqu{q'}{\bar{x}}{y}$
  $A'(\bar{x}, {y}, z) \land C'(\bar{x}, {y}, z)$ such that:
\begin{enumerate}
\item   There is a homomorphism from $q$ to $q'$, i.e., $q\HMPA{\bar{x}}{}q'$;
\item for all queries $q''$ that is homomorphically equivalent to $q$,
  there is an injective homomorphism from $q'$ to $q$, i.e.,
\[
 q\HMPA{\bar{x}}{}q'',  q''\HMPA{\bar{x}}{}q \implies
 q'\HMPA{\bar{x}}{h}q''
\]
in which $h$ is an injective function;
\item  $q'=\extn{q'}$.

\end{enumerate}

Following theorem shows that core of count-distinct queries is unique
up to isomorphically equivalence.
\begin{theorem}\label{thm:core} Given a conjunctive query with comparisons
  $\coqu{q}{\bar{x}}{y}$ $A(\bar{x}, {y}, z) \land C(\bar{x}, {y}, z)$
  if $q_1$ and $q_2$ are cores of $q$ then $q_1 \ISMA{\bar{x}}{} q_2$.
\end{theorem}

\end{definition}

\begin{example}
Consider the following query.
\[
\coqu{q}{\bar{x}}{y}   A(x)\land A(y)  \land A(z) \land A(t)\land x<y
\land y<z \land y<t
\]
$\core{q}$ is 
\[
\coqu{q}{\bar{x}}{y}   A(x)\land A(y)  \land A(z) \land x<y
\land y<z \land x<z
\]
\end{example}

\begin{lemma}\label{thm:complexity} Core of conjunctive queries with comparisons
  can be computed in $\Delta^{\pi}_3$.
\end{lemma}

\begin{lemma}\label{lem:corecomrel} The notion of core of conjunctive
  queries with comparisons is equivalent to the classic notion of core
  conjunctive queries for relational queries.
\end{lemma}

\begin{lemma}\label{lem:equivalenceCQ} Two conjunctive queries with
  comparisons are equal if and only if their cores are isomorphic.
\end{lemma}

\begin{lemma}\label{lem:equivalenceCDCore} If the core of two
  count-distic queries are isomorphic then they are equivalent.
\end{lemma}

\subsection{Equivalence of Count-Distinct  Aggregate Queries}
\subsubsection{Flip of a Query}
Isomorphism of cores of queries is not a necessary condition for
equality of count-distinct queries, and to find a sufficient and
necessary condition we introduce a relaxation of this
condition. First, we introduce the notion of flip of a query which is
a potentially new count-distinct query. We show that existence of
isomorphism between a query or its flip with another query or its flip
is a sufficient and necessary condition for their equivalence.  In
order to define the flip of a query, first we introduce the notion of
equivalence set of a variable in a query as a subset of the variables
of the query containing the initial variable such that all the
variables in the set are homomorphically equivalent to each other
forgetting the comparisons among them.

\begin{definition}[Equal set of a Variable]Given a count-distinct
  query $\coqu{q}{\bar{x}}{y}$ $A(\bar{x}, {y}, z) \land C(\bar{x},
  {y}, z)$, and a variable $v\in \var{q}$, a \emph{potential equal set
    of $v$ w.r.t. $q$} is a set $S\subseteq\var{q}$ such that for all
  variables in $S$ are homomorphically equivalent in the query
  obtained from dropping the comparisons among those of $S$ from
  $q$. More specifically, lets consider $q'$ to be the query obtained
  from $q$ by dropping all comparisons among the variables in $S$. For
  all variables $x_1, x_2$ in $S$:
\[
\left\{
\begin{array}{rl}
q' \HMPA{\sigma}{} q', & \text{ for }  \sigma=\bar{x}x_1\mapsto\bar{x}x_2;\\
q' \HMPA{\sigma'}{} q', & \text{ for } \sigma'=\bar{x}x_2\mapsto\bar{x}x_1.
\end{array}
\right.
\]

\end{definition}

Notice that for each variable $v$, $\{v\}$ is a potential equal
set. Moreover, all the equal sets of a variable are subsets of the
variables of the query.  Now in Lemmas \ref{lem:comparison-char}, \ref
{lem:equalset-charofM}, and \ref{lem:eqalset-char}, and Theorem
\ref{thm:equiv} we prove that each variable has a unique bigest equal
set.

\begin{lemma}\label{lem:comparison-char} Lets consider $C$ to be a satisfiable set
  of comparisons without equality. Followings hold:
\begin{enumerate}
\item For all linearization $L$ of $C$, if $L\models m \le n$ then there
  exists a set of terms $\{r_1,\dots, r_m\}$ such that
\[
(~m~\rho_1~r_1~\rho_2~\dots\rho_{m-1}~r_m~\rho_m~n~) \in C
\]
for $\rho\in\{<, \le\}$. Moreover, if $L\models m < n$, then $\rho_i$
is strict $<$ for at least one $i\in\{1,\dots,m\}$.
\item If there exists a set of terms $\{r_1,\dots, r_m\}$ such that 
\[
(~m~\rho_1~r_1~\rho_2~\dots\rho_{m-1}~r_m~\rho_m~n~)\in C
\]
for $\rho\in\{<, \le\}$, then for all linearization $L$ of $C$,
$L\models m \le n$. Moreover, if $\rho_i$ is strict $<$ for at least one
$i\in\{1,\dots,m\}$, then for all linearization $L$ of $C$, $L\models
m < n$.
\end{enumerate}

\end{lemma}

\begin{lemma} \label{lem:equalset-charofM} Lets consider the
  count-distinct query $\coqu{q}{\bar{x}}{y}$ $ A(\bar{x}, {y},
  z)~\land~ C(\bar{x}, {y}, z)$, a variable $v\in\var{q}$, $M$ to be a
  partial equal set of $v$ w.r.t. $q$, $x, y \in M$, and $C_M$ to be
  the set of comparisons obtained from $C$ by removing those among
  variables in $M$. If for all linearization $L$ of $C(\bar{x}, {y},
  z)$, $L\models x~\rho~r$ for some variable $r\not\in M$, then for
  all linearization $L_M$ of $C_M$, $L_m\models y~\rho~r$.
\end{lemma}
\begin{proof} Since for all linearization $L$ of $C$ $L\models
  x~\rho~r$, by Lemma \ref{lem:comparison-char} there exists a chain
  of comparisons $x~\rho_1~x_1~$ $\rho_2~\dots~$ $\rho_m~r\in$ $C$. Lets
  assume $x_i$ is the last element in this chain that is the member of
  $M$. Considering the fact there is no variable from $M$ in the path
  from $x_i$ to $r$, and using Lemma \ref{lem:comparison-char}, for
  all linearization $L_M$ of $C_M$, $L_M\models x_i\rho r$. Finally,
  given the fact that $x_i$ and $y$ are both in $M$, and by definition
  they are homomorphically equivalent in the query obtained by
  dropping the comparisons among variables of $M$, $L_M\models y\rho
  r$.
\end{proof}

\begin{lemma}\label{lem:eqalset-char} Lets consider a count-distinct query
   $\coqu{q}{\bar{x}}{y}$ $A(\bar{x}, {y}, z) \land C(\bar{x},
  {y}, z)$, a variable $v\in\var{q}$, and an equal set $M$ of $v$
  w.r.t. $q$. For variables $x$ and $y$ in $q$, if there exists a path
  \[
  p: x~\rho_1~x_1~\dots~\rho_m~x_m~\rho_{m+1}y
\] 
in $C(\bar{x}, {y},  z)$, then there exists a path 
\[
x~\rho'_1~x'_1~\dots~  \rho'_n~x'_n~\rho'_{m+1}y
\]
in $C(\bar{x}, {y}, z)$ that does not contain any comparison among the
variables in $M$.
\end{lemma}
\begin{proof} Consider a sub-path
  $x_i~\rho_i~x_{i+1}~\dots~x_k~\rho~n$ of $p$ for which $\{x_i,
  \dots, x_k\}$ to $M$. Using Lemma \ref{lem:equalset-charofM}, we can
  replace such a path by another path between $x_i$ and $n$ such that
  it does not contain any comparison among variables of
  $M$. Consequently, from $p$, we can obtain a path
  $x~\rho'_1~x'_1~\dots~$ $\rho'_n~x'_n$ $~\rho'_{m+1}y$ in $C(\bar{x}, {y},
  z)$ that does not contain any comparison among the variables in $M$.
\end{proof}

Next Theorem shows that there exists a unique maximal equal set for
each variable of a given count-distinct query, partition those
variables. In the rest of the paper, we refer to the bigest equal set
to which a variable belongs wrt $q$ as $\equivv{x}{q}$.

\begin{theorem}\label{thm:equiv}  Given a
  count-distinct query, if a variable has two different equal sets,
  then one of them is a subset of the other.
\end{theorem}
\begin{proof}
Lets assume there are two different equal sets $M$ and
  $N$ to which a variable $n$ belongs. We will show that the set $U =
  M \cup N$ also satisfies all conditions of an equal set, so
  contradicts the assumption that $M$ and $N$ are maximal sets
  satisfying the conditions of equal sets.

  Lets consider the query $q(\bar{x}, \contd(y))\dleq A(\bar{x}, {y},
  z) \land C(\bar{x}, {y}, z)$, and $q_U \dleq$ $A(\bar{x}, {y}, z)
  \land C_U(\bar{x}, {y}, z)$ to be the query obtained
  from $q$ by dropping all comparisons among the variables in
  $U$. Lets consider two variables $x_1, x_2\in U$. We need to show
  that $x_1$ and $x_2$ are homomorphically equivalent wrt $q_U$: i.e.,
  $q_U \homo{\sigma} q_U$, for $\sigma=\bar{x}x_1\mapsto\bar{x}x_2$
  and $q_U \homo{\sigma'} q_U$, for
  $\sigma'=\bar{x}x_2\mapsto\bar{x}x_1$.

One of the followings hold:

\begin{description} 

\item[1- $x_1, x_2\in M$.] Lets consider $q_M \dleq A(\bar{x}, {y}, z)
   \land C_M(\bar{x}, {y}, z)$ to be the query obtained from $q$ by
   dropping all comparisons among the variables in $M$.  Since $x_1$
  and $x_2$ are in $M$, by definition we know that $q_M \homo{\sigma}
  q_M$, for $\sigma=\bar{x}x_1\mapsto\bar{x}x_2$.  That means there
  exist a function $h$ that extends $\sigma$ and is satisfying the
  conditions of the homomorphism from $q_M$ to $q_M$, i.e.,:
\begin{itemize}
\item For all $R(r_1,\dots,r_m)\in q_M$, $R(h(r_1),\dots, h(r_m))\in q_M$;
\item For all linearization $L_M$ of $C_M(\bar{x}, y, \bar{z})$, for all
  comparisons $r~\rho~s\in C_M(\bar{x}, y, \bar{z})$, $L_M\models h(r)~\rho~
  h(s)$, for $\rho\in\{<,>,\le,\ge\}$.
\end{itemize}
Now we show that $h$, is also a witness of the homomorphism
$q_U\homo{\sigma}q_U$. Since $q_M$ and $q_U$ only differ on the
comparisons among the variables, the first condition of the
homomorphism holds, i.e., for all $R(r_1,\dots,r_m)\in q_U$,
$R(h(r_1),\dots, h(r_m))\in q_U$.  So we only need to show that for
all linearization $L_U$ of $C_U(\bar{x}, y, \bar{z})$, for all
comparisons $r~\rho~s\in C_U(\bar{x}, y, \bar{z})$, $L_U\models
h(r)~\rho~ h(s)$. Lets assume $\rho=\le$. For other cases we can prove
similarly. Since for all linearization $L_M$ of $C_M$, $L_M\models
h(r)~\rho~ h(s)$, using Lemma \ref{lem:comparison-char}, we know that
there exists following chain for comparisons
$h(r)~\rho_1~r_1~\rho_2~r_2\dots\rho_m~r_m~\rho_{m+1}h(s)$ in $C_M$,
such that $\rho_i\in\{<, \le\}$. Moreover, using Lemma
\ref{lem:eqalset-char}, we know that there exists a path
$x~\rho'_1~x'_1~\dots~\rho'_n~x'_n~\rho'_{m+1}y$ in $C(\bar{x}, {y},
z)$ that does not contain any comparison among the variables in $M\cup
N$. Consequently, using Lemma \ref{lem:comparison-char}, for all
linearization $L_U$ for $C_U$, $L_U\models h(r)~\rho~ h(s)$.
\item[2- $x_1, x_2\in N$.] Similar to the case 1-.
\item[3- $x_1\in M, x_2\in N$.] Using the case 1-. we can show that
  for $n\in M\cap N$, $n$ and $x_1$, and similarly $n$ and $x_2$ are
  homomorphically equivalent wrt $q_U$. Consequently, $x_1$ and $x_2$
  are homomorphically equivalent.
\item[4- $x_1\in N, x_2\in M$.] Similar to the case 4-.
\end{description}
\end{proof}

\begin{lemma} Given a count-distinct query $\coqu{q}{\bar{x}}{y}$
  $A(\bar{x}, {y}, z) \land C(\bar{x}, {y}, z)$, and a variable $w$,
  $\equivv{x}{q}$ can be computed using Polynomial space in the size
  of the query.
\end{lemma}
\begin{proof} Given a query $q$, we need to guess a subset of
  variables of that satisfies the conditions of an equal set, and
  check if there is no superset of this set that also satisfies the
  conditions of an equal set. This can be done in
  $\Sigma^P_2$. However, checking homomorphism itself is a $\Pi^P_2$
  problem.
\end{proof}

Now we are ready to define the notion of flip of a query.

\begin{definition}[Flip a query] Given a count-distinct query $q$, the
  \emph{flip} of $q$, denoted as $\flip{q}$ is the count-distinct
  query obtained from changing the direction of all comparison
  operators of $\equivv{y}{q}$.
\end{definition}

Notice that by this definition the flip of the query will be equal to
the query itself when $\equivv{y}{q}=\{y\}$.

\begin{lemma} Given a count-distinct query $\coqu{q_1}{\bar{x}}{y}$
  $A(\bar{x}, {y}, z) \land C_1(\bar{x}, {y}, z)$, and its flip 
  $\coqu{q_2}{\bar{x}}{y}$ $A(\bar{x}, {y}, z) \land C_2(\bar{x}, {y},
  z)$, $q_1\ISMA{\sigma}{}q_2$, for $\sigma=\bar{x}\mapsto\bar{x}$.
\end{lemma}

Notice that not necessarily $q_1\ISMA{\sigma'}{}q_2$ for
$\sigma=\bar{x}y\mapsto\bar{x}y$.

\begin{lemma}[Equivalence of the flip-set queries]\label{lem:equiv-flip} Consider the
  count-distinct query $q(\bar{x})\dleq A(\bar{x}, \bar{y}) \land
  C(\bar{x}, \bar{y})$, and $q'=\flip{q}$.  $q=q'$.
\end{lemma}
\begin{proof}[idea] By definition $q'$ has been obtained from flipping
  the comparisons of the equal set of $y$.  First
  notice that if $|\equivv{y}{q}|>1$, then $\equivv{y}{q}$ cannot
  contain any variable from $\bar{x}$. For
  all variables $x_1$ and $x_2$ in the equivalence class
  $\equivv{y}{q}$, the existence of homomorphisms
  $\homo{\bar{x}\mapsto\bar{x}, x_1\mapsto x_2}$ and
  $\homo{\bar{x}\mapsto\bar{x}, x_2\mapsto x_2}$ guarantees that they
  match exactly the same set of terms for each group by value
  $\bar{x}\ra\bar{d}$ in answering $q'$ over all databases. Since
  $\equivv{x}{q}$ does not contain any variable from the group by
  variables, for the purpose of equivalence and containment of the
  queries it's only important to check for all assignments $\bar{d}$
  to $\bar{x}$ they match exactly same number of items in $q$.

\end{proof}

\subsubsection{Equivalence of Count Distinct Queries}

\begin{lemma}\label{lem:comp-answer} Given a set of variables $M$, a variable $y$ in $M$, a
  set $C$ of comparisons among variables of $M$, and a set of ordered
  values $D$. Lets consider the DAG constructed out of $C$ in which
  the variables are the nodes, and there is a and edge with weight $1$
  from $x$ to $y$ if $X>y$ and an edge with weight $0$ if $x\ge 0$.
  The number of different values that can be assigned to $y$ through
  assignments of values of $D$ to the variables of $M$ that satisfy
  all comparisons of $C$ is $|D| -
  (M_h+M_l)$, in which $M_h$ is the maximum distance of the variable
  $y$ from one of its parents, and $M_l$ is the maximum distance of
  $y$ from one of its children in the DAG representing the comparisons
  of $D$.
\end{lemma}

\begin{lemma}\label{lem:flip-char} Lets consider two sets of variables
  $M$ and $M'$, variables $y\in M$ and $y'\in M'$, and set $C$ ($C'$)
  of comparisons among variables of $M$ ($M'$). 
  Lets consider the DAG $D$ ($D'$)constructed out of $C$ ($C'$) in
  which the variables are the nodes, and there is a and edge with
  weight $1$ from $x$ to $y$ if $X>y$ and an edge with weight $0$ if
  $x\ge 0$.
  With $\ans{y, C, D}$ we denote the number of different assignments
  of values of $D$ to the variables in $C$ that satisfy all
  comparisons of $C$.

  For all set $D$ of values from an ordered domain, $\ans{y, C,
    D}\le\ans{y', C', D}$ if and only if one of followings
  hold:
\begin{itemize}
\item There exists a homomorphism $h$ from DAG of $C'$ to DAG of $C$ such
  that for all linearization $L$ of the comparisons of $C$, for all
  comparison $r~\rho~s \in C'$, $L\models h(r) \rho h(s)$;
\item Lets denote the set of comparisons obtained from $C'$ by
  flipping their direction with $C''$. There exists a homomorphism
  $h''$ from DAG of $C''$ to DAG of $C$ such that for all
  linearization $L$ of the comparisons of $C$, for all comparison
  $r~\rho~s \in C''$, $L\models h(r) \rho h(s)$;
\end{itemize}
\end{lemma}
\begin{proof} Direct consequence of the Lemma \ref{lem:comp-answer}.
\end{proof}

\begin{theorem} Given two count distinct queries, it is decidable to
  check if they are equivalent.
\end{theorem}
\begin{proof}[idea] The algorithm to check the equivalence of
  count-distinct queries is as follows. Lets consider the two queries
  to be compared are $\coqu{q_1}{\bar{x}}{y}$ $A_1(\bar{x}, {y}, z)
  \land C_1(\bar{x}, {y}, z)$ and $\coqu{q_2}{\bar{x}}{y}$
  $A_2(\bar{x}, {y}, z) \land C_2(\bar{x}, {y}, z)$,

\begin{itemize}

\item Compute $q'_1=\core{q_1} \text{ and } q'_2=\core{q_2}$;
\item Compute $q'_{1f}=\flip{q_1} \text{ and } q'_{2f}=\core{q_2}$;
\item If for all $q"_1\in \{q'_1, q'_{1f}\}$ and for all $q"_2\in
  \{q'_2, q'_{2f}\}$ $q"_1\not \ISMA{\bar{x}}{}q"_2$ then $q_1$ and $q_2$
  are not equal; otherwise, they are equal.
\end{itemize}

We need to show that that the isomorphism between the core or flip
of the cores of the queries is a sufficient and necessary condition
for deciding their equivalence. Using Lemma \ref{lem:equivalenceCQ},
we know that the $q_1$ is equal to $q'_1$ and $q_2$ is equal to
$q'_2$. Lemma \ref{lem:equiv-flip}, implies that $q'_{1f}$ is equal to
$q'_{1}$ and $q'_{2f}$ is equal to $q'_{2}$. As a result if one of
$q'_1$ or $q'_{1f}$ is isomorphic to one $q'_{2}$ ot $q'_{2f}$ then we
can infer that $q_1$ and $q_2$ are also equal.

\end{proof}


\section{Conclusion}
In this paper we discussed the problem of equivalence of
count-distinct queries with comparisons, and proved the decidability
of the problem through introducing two new notions of core of queries
with comparisons and their flip over dense ordered domains. It is
still open what will happen over discrete domains. Since we address
the problem of equality through checking isomorphism, another
legitimate open question is what will happen to the problem of
containment of count-distinct queries.5


\bibliographystyle{abbrv}
\bibliography{main-bib}

\end{document}